\documentclass{ws-procs975x65}
\begin{document}


\def\makarovsmirnovicmpsle#1{SLE$\left(#1\right)$}
\def\makarovsmirnovicmpbr#1{\left(#1\right)}
\def\makarovsmirnovicmpbrs#1{\left\{#1\right\}}
\newtheorem{makarovsmirnovproblem}{Problem}
\renewcommand{\themakarovsmirnovproblem}{\Alph{makarovsmirnovproblem}}

\title{Off-critical lattice models and massive SLEs}

\author{N. MAKAROV}

\address{Mathematics 253-37, Caltech\\
Pasadena, CA 91125, USA\\
makarov@caltech.edu}

\author{S. SMIRNOV}

\address{Section de math\'ematiques, 
Universit\'e de Gen\`eve\\
2-4, rue du Li\`evre, c.p. 64,   
1211 Gen\`eve 4, Switzerland\\
E-mail: stanislav.smirnov@math.unige.ch}

\begin{abstract}
We suggest how versions of Schramm's SLE
can be used to describe the scaling limit of 
some off-critical 2D lattice models.
Many open questions remain.
\end{abstract}


\bodymatter

\section{Introduction}

During the last 25 years
Conformal Field Theory (CFT) was successful in 
heuristically describing 
conformally invariant scaling limits
of 2D lattice models at criticality, such as
the Ising model, percolation, Self-Avoiding Polymers, Potts models,
--- see Ref. \refcite{cft-collection} for a collection of the founding papers of the subject.

Recently there was much progress in the mathematical understanding,
in large part due to Oded Schramm's introduction of {SLE}{s},
or \emph{Schramm Loewner Evolutions}.
\makarovsmirnovicmpsle{\kappa}~is a one-parameter family of random conformally invariant curves,
constructed by running a Loewner Evolution with (real valued) Brownian motion 
as the driving term.
For several models convergence to {SLE}~was established in the scaling limit;
moreover, {SLE}~is well-adapted to calculations, which typically boil down to 
It\^o's stochastic calculus.
See Ref. \refcite{schramm-icm,smirnov-icm1} for an exposition and references.

The key property of {SLE}~is its conformal invariance, 
which is expected in 2D lattice models only at criticality,
and the question naturally arises:

\centerline{\emph{Can {SLE}~success be replicated for off-critical models?}}

\noindent
In most off-critical cases to obtain a non-trivial scaling limit one has to adjust
some parameter (like temperature in the Ising model or probability of an open site in percolation),
sending it at an appropriate speed to the critical value.
Such limits lead to massive field theories, 
so the question can be reformulated as whether one can use {SLE}{s} to describe those.
Massive CFTs are no longer conformally invariant, but are still
covariant when mass is considered as a variable covariant density.

So far there was limited progress on off-critical {SLE}{s}, 
related work limited to
Ref. \refcite{bauer-bernard-kytola,bauer-bernard-cantini,nolin-werner,garban-pete-schramm,makarov-smirnov-mass}.
Below we propose an approach based on the combination of
{SLE}~\emph{Martingale Observables} (MO) with potential theory and stochastic analysis,
and we start by describing the critical case.

\subsection{SLE~and  critical lattice models}

Suppose that a critical lattice model defines for every simply connected domain $\Omega$
with two marked boundary points $a$, $b$ a random discrete simple curve 
joining them inside $\Omega$.
One can think e.g. of a domain wall boundary in the Ising model with Dobrushin boundary conditions 
or of a LERW -- the Loop Erasure of the Random Walk from $a$ to $\partial\Omega\setminus\makarovsmirnovicmpbrs{b}$
conditioned on ending at $b$.

The key observation made by Oded Schramm was that if 
a conformally invariant scaling limit $\gamma$
of the discrete curves  exists 
and satisfies Markov property (i.e. the curve progressively drawn from $a$ to $b$
does not distinguish its past from the boudary of the domain $\Omega$),
then it can be described by \makarovsmirnovicmpsle{\kappa} for some $\kappa\in[0,\infty[$.
Take an appropriate time parameterization $\gamma(t)$
and denote by $\Omega_t$ the component at $b$ of the domain $\Omega\setminus\gamma[0,t]$.
Then the random Loewner conformal map 
$$Z_t(z):\,\Omega_t\to{\mathbb C}_+,~~~\gamma(t)\mapsto0,~b\mapsto\infty,~\text{normalized~at}~ b$$
(in what we call a \emph{Loewner chart}),
satisfies the \emph{Loewner equation} with the Brownian Motion as the driving term.
We write it as a \emph{stochastic differential equation} (SDE)
\begin{equation}\label{eq:sle}
d Z_t(z)=\frac2{Z_t(z)}dt -d\xi_t,~~
d\xi_t:=\sqrt\kappa d B_t,
\end{equation}
with the SLE \emph{driving term}
$\xi_t$ given by the standard Brownian motion $B_t$.

Moreover, it turns out that to deduce {SLE}~ convergence it is enough to show
that just one \emph{observable}
$M_t(z)=M(z,\Omega_t)$
(e.g. spin correlation or percolation probability)
is conformally covariant and Markov in the limit, see Ref. \refcite{smirnov-icm1} for a discussion.
So far conformal invariance of observables was always estbalished by showing that they
are discrete holomorphic (or harmonic) functions of a point satisfying
some \emph{Boundary Value Problem} (BVP) -- Dirichlet, Neumann or Riemann-Hilbert.

The covariant holomorphic MOs for {SLE}~can be easily classified, leading to
a one parameter family for each $\kappa$ or each \emph{spin} (conformal dimension) $\sigma$,
namely
\begin{equation}\label{eq:mo}
M_t^{\kappa,\sigma,\beta}(z)=M_t(z)=Z_t(z)^\beta Z_t'(z)^\sigma,~~~\text{with}~\sigma=\beta+\frac{\beta(\beta-1)}4\kappa,
\end{equation}
is $\makarovsmirnovicmpbr{dz}^\sigma$-covariant.
It is characterized (see Ref.~\refcite{smirnov-icm1}) by the Riemann-Hilbert BVP
\begin{equation}\label{eq:mobvp}
\bar\partial M_t=0~\text{in}~\Omega_t,~~M_t(z)\parallel \tau^{-\sigma}~\text{on}~\partial\Omega_t,~~\text{appropriate~singularities~at}~\gamma(t),\,b.
\end{equation}
Note that the MO above can also be rewritten in terms
of the complex Poisson kernel in $\Omega_t$ at $\gamma(t)$,
$P_t(z):=-1/Z_t(z)$.

Of special interest is also the family of holomorphic MOs  for \makarovsmirnovicmpsle\kappa:
\begin{equation}\label{eq:bmo}
M_t^\kappa(z)=M_t(z)=\log Z_t(z)+\makarovsmirnovicmpbr{1-\frac\kappa4}\log Z_t'(z).
\end{equation}
Those are covariant pre-pre-Schwarzian forms, and so
it suffices to study their imaginary parts --  harmonic MOs characterized by
\begin{equation}\label{eq:bmobvp}
\Delta M_t(z)=0~\text{in}~\Omega,~~M_t(z)=\arg Z_t(z)+\makarovsmirnovicmpbr{1-\frac\kappa4}\arg Z_t'(z)~\text{on}~\partial\Omega_t,
\end{equation}
the Dirichlet BVP being well posed and independent of time parameterization,
i.e. the choice of normalization at $b$ of the map $Z_t(z)$.
Such \emph{bosonic} observables 
(in Coulomb gas formalism of CFT, they are 1-point functions of 
the bosonic field in the presence of background charge)
feature prominently in 
Ref.~\refcite{schramm-sheffield-harmonic,schramm-sheffield-dgff,dubedat-partition}.

The observables above have the martingale property with respect to corresponding \makarovsmirnovicmpsle{\kappa}
by a simple application of It\^o's calculus:
vanishing of their drifts under the diffusion (\ref{eq:sle})
is easy to deduce using (\ref{eq:sle}) and its corollary
\begin{equation}\label{eq:led}
d Z_t'(z)=-\frac{2Z_t'(z)}{Z_t(z)^2}dt.
\end{equation}
Conversely, if a random curve admits a MO of the mentioned form,
similar calculations prove the curve to be a \makarovsmirnovicmpsle{\kappa}.

\subsection{SLE and off-critical perturbations}

Most discrete holomorphic observables 
studied in Ref. \refcite{MR2044671,schramm-sheffield-harmonic,
smirnov-fk1,riva-cardy-hol,schramm-sheffield-dgff,chelkak-smirnov-iso, chelkak-smirnov-spin} 
behave well under some off-critical
perturbations, leading to discrete \emph{Massive Martingale Observable} (MMO) 
satisfying the massive version
of the Cauchy-Riemann or Laplace equations:
\begin{align}\label{eq:mcr}
\bar\partial^{(m)} M^{(m)} &:=\bar\partial M^{(m)} - i m \overline{M^{(m)}}=0,\\
\label{eq:mlap}
\Delta^{(m)} M^{(m)} &:=\Delta M^{(m)} - m^2 {M^{(m)}}=0,
\end{align}
inside $\Omega$ and solving the same boundary value problem as the critical MO.
Here mass $m$ can be understood as a function of $z$,
changing covariantly under conformal transformations.
In discrete setting a power of the lattice mesh enters Eq.~(\ref{eq:mcr},\ref{eq:mlap}),
thus one has to tend the perturbation parameter to its critical value
in a coordinated way with the lattice mesh, so that the mass does not blow up
in the scaling limit.

Given a discrete random curve with a discrete MMO,  
we can ask, to what extent the {SLE}~ theory can be applied in the massive case,
posing the following problems:
\begin{makarovsmirnovproblem}\label{p1}
Show that discrete MMO has a scaling limit, which is then a MMO for some random curve.
Find a SDE for the driving diffusion $d\xi_t$ replacing the Brownian motion
in the Loewner Evolution (\ref{eq:sle}).
\end{makarovsmirnovproblem}
\begin{makarovsmirnovproblem}\label{p2}
Show that this SDE has a unique solution and the corresponding random curve
is the scaling limit of the original discrete curves.
\end{makarovsmirnovproblem}
\begin{makarovsmirnovproblem}\label{p3}
Use massive {SLE}{s} to derive properties of massive field theories and
massive {SLE}~curves.
\end{makarovsmirnovproblem}
Note that massive models (as well as massive holomorphic functions)
usually are even easier to control, so the main problem is the absence
of conformal invariance, or rather presence of conformal covariance with respect to the mass.
Consequently the drift terms in the corresponding diffusions depend on (Euclidean) geometry
of the domains constructed dynamically by Loewner evolution, 
leading to SDEs with general previsible path functionals, which are rather complicated.

\paragraph{Outline of the paper.}
We were able to advance within the suggested framework,
and in Section~\ref{sec:thm} we  present some of our results
from the forthcoming Ref. \refcite{makarov-smirnov-mass}.
In Section~\ref{sec:proofs} we provide two essential elements of the proof, 
namely we give an example of the drift computation (in the case of bosonic observables) 
and establish some simple a priory estimates of the drift (in the $\kappa=4$ bosonic case). 
Finally we conclude with a list of open questions in Section~\ref{sec:quest}.

\paragraph{Notation.}
We consider a simply connected domain $\Omega\subset{\mathbb C}$ 
with two marked 
points $a, b\in\partial\Omega$,
joined by a random Markov curve $\gamma(t)$ with Loewner parameterization. 
The Loewner map to the half-plane is denoted by $Z_t(z)$.
We denote by $h$ the harmonic measure, by $G$ the Green's function, and its boundary differentials by
$$P(u,z)= N_u G(\cdot,z),\qquad K(u,b)= N_b P(u,\cdot)$$
(the Poisson  kernel  and Poisson boundary  kernel respectively).
Here $N_u$ stands for  the normal derivative at $u$
in the Loewner boundary chart.
By $G^{(m)},P^{(m)},\dots$ we denote the 
corresponding massive objects.
The index $t$ signifies that we work in the domain $\Omega_t$,
i.e. the component at $b$ of $\Omega\setminus\gamma[0,t]$,
and $\iint$ denotes the area integral.

\section{Results}\label{sec:thm}

Below we go through all the cases discovered to-date where a lattice model
admits a discrete MO with a massive perturbation,
and describe results which will appear in Ref.~\refcite{makarov-smirnov-mass}.
For simplicity we restrict ourselves to the case of constant mass $m$,
but the methods should apply to an appropriate class of variable densities $m(z)$.

\subsection{Loop Erased Random Walk with killing, $\kappa=2$}\label{sub:lerw} 

The LERW converges to \makarovsmirnovicmpsle2, as shown in Ref. \refcite{MR2044671}.
The law $\gamma(t)$ of the LERW${}^{(m)}$ on a lattice of mesh $\epsilon$ is defined  by applying the same
loop erasing  procedure to the random walk with a killing rate 
(i.e. probability to die out at each step) 
$$\delta=m^2\epsilon^2.$$
The RW is done inside a domain $\Omega$ from the boundary point $a$ to the rest of the boundary and conditioned on
ending at $b\in\partial\Omega$ -- this ensures the Markov property.

\begin{theorem} For a bounded domain the scaling limit of  LERW${}^{(m)}$ exists and is given by the 
massive \makarovsmirnovicmpsle2 with the driving diffusion
\begin{equation}\label{eq:lerw}
\xi_t=\sqrt2dB_t+\lambda_tdt,~~~\lambda_t=2\left[\log K^{(m)}(\cdot,b; \Omega_t)\right]'(\gamma_t)
~~({\rm in~the~chart}~Z_t).
\end{equation}
The law of the scaling limit  is absolutely continuous with respect to \makarovsmirnovicmpsle2.
\end{theorem}
The proof begins by constructing a discrete harmonic MMO
(on appropriate discretization of the domain $\Omega$):
\begin{equation*}
M_t^{(\delta)}(z)\equiv M^{(\delta)}(z;\gamma(t),b,\Omega_t)
:=\frac{G^{(\delta)}(\gamma(t),z;\partial\Omega\cap\gamma[0,t[\,)}
{h^{(\delta)}(\gamma(t),b; \partial\Omega\cap\gamma[0,t[\,)},\qquad z\in \Omega_t^{\text{discrete}},
\end{equation*}
where $G^{(\delta)}$ and $h^{(\delta)}$  denote discrete Green's function and discrete harmonic measure respectively, 
both with killing rate $\delta$. 
Then we we fix a boundary chart at $b$ and
show that after an appropriate normalizations the discrete MMO converges 
to 
\begin{equation}\label{eq:lerwmo}
\frac{P_t^{(m)}(\gamma(t),z)}{K_t^{(m)}(\gamma(t),b)},\qquad (z\in \Omega_t),
\end{equation}
a continuous MMO 
corresponding to (\ref{eq:mo},\ref{eq:mobvp}) with $\beta=-1$ and $\sigma=0$.
Now let the driving term $\xi_t$ be an It\^o process 
$$d\xi_t=\lambda(t,\omega)~dt+\sigma(t,\omega)~dB_t.$$
We claim that if the Loewner chain has MMO (\ref{eq:lerwmo}),
then $\sigma\equiv\sqrt2$ and $\lambda$ satisfies (\ref{eq:lerw}). 
We now have the equation
$$d\xi_t=\lambda(t,\xi_\bullet)~dt+\sqrt2~dB_t,$$
with the \emph{previsible} path functional $\lambda(t,\xi_\bullet)$ given by (\ref{eq:lerw}). 
We claim that this equation satisfies the standard finiteness  condition 
(cf. Chapter 5 in Ref.~\refcite{rogers-williams-2}):
$\lambda$ is locally integrable on almost all paths. 
The question of course arises of  existence and uniqueness of SDE solutions. 
In the case under consideration the answer is quite  simple:
the path functional $\lambda$ satisfies the Novikov's condition
(see 
Ref. \refcite{oksendal})
\begin{equation}\label{eq:novikov}
{\mathbb E}\left[\exp\makarovsmirnovicmpbr{\frac12\int_0^\infty\lambda(t,\xi_\bullet)^2\,dt}\right]<\infty,
\end{equation}
(and in fact,$\int_0^\infty \lambda(t,\xi_\bullet)^2~dt\le\mathrm{const} <\infty$ holds for all paths).
In particular, the SDE has a unique in law (weak) solution $\xi_t$
and its law is absolutely continuous with respect to $\sqrt 2 dB_t$. 
See Proposition~\ref{prop:novikov} below for a similar argument.
One can then deduce that the massive {SLE}~is the scaling limit of the massive LERW by using
the absolute continuity of latter with respect to the LERW.

\subsection{Massive Harmonic Explorer and 
Gaussian Free Field, $\kappa=4$ } 

We use the original Harmonic Explorer (HE) construction 
on a hexagonal lattice, Ref.~\refcite{schramm-sheffield-harmonic}. 
At each step, the boundary consists of two  arcs, 
and the explorer turns in the direction of one of the arcs with probability equal to 
its harmonic measure in the current domain $\Omega_t$ evaluated at the ``growth point.''
Introducing the killing rate as in the previous section, we get one distinction 
from the massless case --
the two harmonic measures don't sum up to one, so with complementary probability we toss a fair coin to determine the direction of the turn. 
The MMO is the massive version of the bosonic MO (\ref{eq:bmo},\ref{eq:bmobvp}) for $\kappa=4$.

The theory is completely analogous to that of the massive LERW.
If the initial domain is bounded, then the scaling limit exists 
and is absolutely continuous with respect to \makarovsmirnovicmpsle4, the scaling limit of massless explorer. 
We deduce the formula (\ref{eq:drift}) for the drift in more general bosonic case in Proposition~\ref{prop:drift}
(for the HE $M_t$ is the difference of the  \emph{massless} harmonic measures of the two boundary arcs). 

The same framework holds for the massive version of the discrete Gaussian Free Field
discussed in Ref.~\refcite{schramm-sheffield-dgff},
and parts of our construction apply to  general bosonic observables.

\subsection{Massive Peano curves, $\kappa=8$}  

For a lattice approximation of $\Omega$,
we choose  its cycle-free subgraph $\Gamma$
with probability proportional to $\alpha^{n(\Gamma)}$, $n(\Gamma)$ being the number of 
connected components.
Conditioning $\Gamma$ to contain all the edges in the boundary arc $(a,b)$
creates a random Markov interface $\gamma$  
from $a$ to $b$ which traces the component wired on the arc $(a,b)$.
The case $\alpha=0$ corresponds to the usual Uniform Spanning Tree model,
as considered in Ref.~\refcite{MR2044671}, whose
interface converges to the random Peano curve \makarovsmirnovicmpsle8.

The full massive harmonic measure of the boundary with reflection  
in the unwired part is a discrete MMO, corresponding to (\ref{eq:mo},\ref{eq:mobvp})
with $\beta=1/2$ and $\sigma=0$.
It has a scaling limit, and we show that the drift has to be 
\begin{equation}\label{eq:driftust}
d\lambda_t=16\makarovsmirnovicmpbr{\iint_{\Omega_t} \tilde P_t^{\phantom{(}}\tilde P_t^{(m)}}dt,
\end{equation}
where $\tilde P_t$ is the minimal (Martin's) kernel for Dirichlet/Neumann boundary conditions in 
$(\Omega_t, \gamma(t), b)$ and $\tilde P_t^{(m)}$  its massive counterpart 
in the Loewner chart. 

It would be interesting to interpret the formal expression (\ref{eq:driftust})
and show that the 
SDE is well defined.
Note that for $\alpha>0$ the interface is no longer space-filling,
thus massive {SLE} and {SLE}~are mutually singular,
while both have scaling dimension $2$.

\subsection{Fortuin-Kasteleyn Ising model, $\kappa=16/3$}\label{sub:fk}

The \emph{fermionic} MO, considered in 
Ref.~\refcite{smirnov-icm1,smirnov-fk1,smirnov-fk2,chelkak-smirnov-iso,riva-cardy-hol}
for the random cluster representation of the critical Ising model,
implies that the interface converges to \makarovsmirnovicmpsle{16/3}. 
The MO corresponds to (\ref{eq:mo},\ref{eq:mobvp}) with $\beta=-1/2$ and $\sigma=1/2$,
and becomes a MMO under the perturbation by $p$ --
the weight of an open edge (FK analogue of magnetization).
The techniques of  Ref.~\refcite{chelkak-smirnov-iso} allow
to show the existence of a scaling limit MMO.
It solves a Riemann-Hilbert boundary value problem,
which makes potential theory and hence derivation of drifts
more difficult than in the bosonic case.

\subsection{Critical Ising model, $\kappa=3$}\label{sub:spin}

A similar fermionic MO appears 
(see Ref. \refcite{smirnov-icm1,chelkak-smirnov-spin})
in the usual spin representation of the Ising model
and  implies convergence of the interface (domain wall boundary) to \makarovsmirnovicmpsle3.
This is observable (\ref{eq:mo},\ref{eq:mobvp}) with $\beta=-1$ and $\sigma=1/2$.
When perturbed by the energy field, it becomes a MMO,
which again solves a Riemann-Hilber boundary value problem.
While the techniques of  Ref.~\refcite{chelkak-smirnov-iso} 
should also be applicable, there are similar difficulties as well.

\section{Techniques}\label{sec:proofs}

As an example, we show how to derive the drift and  analyze  the corresponding SDE 
for bosonic MOs (\ref{eq:bmo}). 
We work with their imaginary parts, harmonic MOs (\ref{eq:bmobvp}).
The corresponding MMOs are massive harmonic (\ref{eq:mlap}),
while solving the same BVPs.

\subsection{Deriving the drift}

\begin{proposition}\label{prop:drift}
If a random curve is described by a Loewner evolution and has a  bosonic MMO satisfying (\ref{eq:dbmo}),
then the driving diffusion is given by
\begin{equation}\label{eq:drift}
d\xi_t=\sqrt\kappa dB_t+d\lambda_t, ~~~ d\lambda_t  = \makarovsmirnovicmpbr{\iint_{\Omega_t} m^2 M^{\phantom{(}}_t P^{(m)}_t}dt.
\end{equation}
\end{proposition}
\begin{remark}\label{rem:drift} 
The drift can be rewritten in several ways, e.g. as $\iint_{\Omega_t} m^2 M^{(m)}_t P^{\phantom{(}}_t$.
\end{remark}
\begin{proof}
Let $M^{(m)}_t$ be the massive harmonic MMO solving the BVP (\ref{eq:dbmo}), then
\begin{equation}\label{eq:mm}
\Delta^{(m)} M_t^{(m)}= 0\Rightarrow\Delta^{(m)}(M_t^{(m)}-M^{\phantom{(}}_t) 
= m^2 M^{\phantom{(}}_t\Rightarrow M_t^{(m)} = M^{\phantom{(}}_t+m^2 M^{\phantom{(}}_t*G_t^{(m)},
\end{equation}
Similarly for the massive Poisson kernel we have
\begin{equation}\label{eq:mp} 
P_t^{(m)}=P^{\phantom{(}}_t+m^2 P^{\phantom{(}}_t*G_t^{(m)}.
\end{equation}
Describe the massive curve by the Loewner evolution
with the driving term $\xi_t$ 
$$d \xi_t= \sqrt{\kappa}d B_t +d\lambda_t.$$
Denoting drifts with respect to {SLE}~ and massive {SLE}~ by $d^{sle}$ and $d^{msle}$, 
we calculate
\begin{equation}\label{eq:dm}
d^{msle} M^{\phantom{(}}_t= d^{sle} M^{\phantom{(}}_t + {\mathrm{Im}}\makarovsmirnovicmpbr{{d\lambda_t}/{Z_t}}=P^{\phantom{(}}_t d\lambda_t.
\end{equation}
Using the massive version of the Hadamard's variational formula
\begin{equation}\label{eq:hadamard}
d^{msle} G_t^{(m)}(z,w)= - P_t^{(m)}(z) P_t^{(m)}(w)dt,
\end{equation}
we can write the drift for the MMO:
\begin{align*}
0~=\:\:\:\,&~d^{msle} M_t^{(m)}(z)~\overset{(\ref{eq:mm})}{=}
~d^{msle}\makarovsmirnovicmpbr{M^{\phantom{(}}_t+m^2 M^{\phantom{(}}_t*G_t^{(m)}}\\
~=\:\:\,\,&~d^{msle} M^{\phantom{(}}_t + m^2(d^{msle} M^{\phantom{(}}_t)*G_t^{(m)} 
+ m^2 M^{\phantom{(}}_t(\cdot) * (d^{msle} G_t^{(m)}(z,\cdot))\\
\overset{(\ref{eq:dm},\ref{eq:hadamard})}{=}
&~\makarovsmirnovicmpbr{P^{\phantom{(}}_t+m^2 P^{\phantom{(}}_t*G_t^{(m)}}d\lambda_t - m^2 M^{\phantom{(}}_t(\cdot) 
* (P_t^{(m)}(z) P_t^{(m)}(\cdot))dt\\
~\overset{(\ref{eq:mp})}{=}~&~P_t^{(m)}(z)d\lambda_t  - P_t^{(m)}(z) 
\makarovsmirnovicmpbr{\iint_{\Omega_t} m^2 M^{\phantom{(}}_t P_t^{(m)}}dt, 
\end{align*}
so we deduce (\ref{eq:drift}).
\end{proof}

\subsection{Analysis of diffusion}
\begin{proposition}\label{prop:novikov}
In bounded domain $\Omega$  the massive \makarovsmirnovicmpsle4 driven by the diffusion (\ref{eq:drift})
for $\kappa=4$
is  absolutely continuous
with respect to \makarovsmirnovicmpsle4.
\end{proposition}
\begin{proof}
First recall that in this case the harmonic MMO  is bounded by the maximum principle since
$M_t^{(m)}=\arg Z_t\in[0,\pi]$ on $\partial\Omega_t$ by (\ref{eq:dbmo}). Thus by Remark~\ref{rem:drift} 
\begin{equation}\label{eq:lest}
|\lambda_t|=\iint_{\Omega_t}m^2 M^{m}_t P^{\phantom{(}}_t\lesssim \iint_{\Omega_t} P_t.
\end{equation}
Now for a bounded initial domain we can write using the Hadamard's variational formula (\ref{eq:hadamard}):
\begin{align*}
\int_0^\infty\lambda^2_t~dt&\overset{(\ref{eq:lest})}{\le}
\int_0^\infty dt\iint_{(z)} P_t(z) \iint_{(w)} P_t(w)
\overset{(\ref{eq:hadamard})}{=}-\int_0^\infty \iint_{(z)}\iint_{(w)} d G_t(z,w)\\
&=\iint_{(z)}\iint_{(w)}\left[G_0(z,w)-G_\infty(z,w)\right]\le \iint_{(w)}\iint_{(z)} G_0(z,w)\le C<\infty.
\end{align*}
The last step follows since $\Omega_0$ 
is bounded and so
$$G_{\Omega_0}(z,w)\le -\log{|z-w|}+\mathrm{const}.$$
Applying the Novikov's criterion (\ref{eq:novikov}),
we conclude that our diffusion
is well defined and its law is absolutely continuous 
with respect to \makarovsmirnovicmpsle{4}.
\end{proof}

\section{Questions}\label{sec:quest}

We would like to end with a few questions,
which originated in our work.

\begin{question}\label{q1}
How many ``physically interesting'' (e.g. relevant for the renormalization group)
massive perturbations can a CFT have?
The arguments above display a one-parameter family of covariant holomorphic
MO for each \makarovsmirnovicmpsle{\kappa},
suggesting that somehow there is a one-parameter family of ``canonical'' 
massive perturbations.
Are perturbations by non-holomorphic MO also relevant? 
Are all perturbations generated by MO perturbations? 
\end{question}

\begin{question}\label{q2}
Is it always true that diffusion driving a massive version of 
\makarovsmirnovicmpsle\kappa~is a speed $\kappa$ Brownian motion plus a drift?
In the absence of conformal invariance the drift forcibly depends on 
the geometry of the domain grown. 
Are all such drifts locally bounded variation path functionals? 
\end{question}

\begin{question}\label{q3}
The drifts so far encountered are either absolutely continuous 
(e.g. LERW and HE cases above)
or monotone (e.g. UST case above or percolation perturbations 
discussed 
in Ref.~\refcite{nolin-werner,garban-pete-schramm}).
Are all the possible drifts combinations of those?
\end{question}

\begin{question}\label{q4}
Can one show that a {SLE}~ and its massive version are
always in the same universality class
(in the sense of scaling exponents)?
\end{question}

\begin{question}\label{q5}
For $\kappa\le4$ in our examples the massive and usual {SLE}~are mutually absolutely continuous. 
Is it true for all perturbations when $\kappa\le4$?
\end{question}

\begin{question}\label{q6}
Is it it true that for $\kappa>4$ the massive and usual {SLE}~ are singular, 
while being in the same universality class?
If not, for which perturbations are they mutually absolutely continuous?
We expect the bosonic ones to be among those.
\end{question}

\begin{question}\label{q7}
\makarovsmirnovicmpsle\kappa~almost surely produce simple curves for $\kappa\le4$ and curves with double points for
$\kappa>4$. ``Resampling'' the model at a double point alters the curve drastically.
Is this the reason for the absolute continuity / singularity dichotomy suggested above?
\end{question}

\begin{question}\label{q9}
Does discrete percolation observable from Ref. \refcite{0985.60090}
have a massive counterpart? If so, which perturbation
does it correspond to?
This observable requires three marked points, but when two are fused,
its continuous counterpart becomes the observable (\ref{eq:mo},\ref{eq:mobvp})
with $\beta=1/3$ and $\sigma=0$.
\end{question}

\begin{question}\label{q10}
We have MMOs for at least five different values of $\kappa$.
Guess by analogy the driving diffusions for other values of $\kappa$,
and show that those are well-defined and lead to random curves.
\end{question}

\begin{question}\label{q11}
In particular,  FK Ising MMO from section~\ref{sub:fk}
belongs to a family of MMOs arising from (\ref{eq:mo},\ref{eq:mobvp}) with  $\sigma=-\beta=-1+8/\kappa$,
cf. Ref. \refcite{smirnov-icm1}.
We believe that all those correspond to magnetization perturbations.
Is it true, in particular for $\kappa=3$ (spin Ising)?
If it holds for $\kappa=6$, can one make a connection with a heuristic
formula $d\lambda_t=\mathrm{const} |d\gamma(t)|^{3/4} |dt|^{1/2}$
for the drift from Ref.~\refcite{garban-pete-schramm}?
\end{question}

\begin{question}\label{q12}
Similarly, starting from the spin Ising MMO from section~\ref{sub:spin},
we can ask (cf. Ref. \refcite{smirnov-icm1}) whether
all MMOs arising from (\ref{eq:mo},\ref{eq:mobvp}) with $\sigma=-\beta/2=3/\kappa-1/2$
correspond to fugacity perturbations?
Note that while LERW MMO from section~\ref{sub:lerw} does not belong to this family,
its differential does.
\end{question}

\begin{question}\label{q8}
The $O(N)$ model has conjecturally two critical regimes,
corresponding to parameter $x$ 
(bond weight in the loop representation of the high temperature expansion)
belonging to $\makarovsmirnovicmpbrs{x_c}$ and $]x_c,\infty[$, 
see Ref. \refcite{kager-nienhuis}.
Interfaces conjecturally converge to \makarovsmirnovicmpsle{\kappa} and \makarovsmirnovicmpsle{\tilde\kappa} correspondingly,
where 
$\kappa$ and $\tilde\kappa$ are known functions of $N$ and satisfy duality 
(though different from the usual duality $\kappa\kappa^*=16$):
\begin{equation}\label{eq:dua}
\frac1\kappa+\frac1{\tilde\kappa}=\frac12,~~\kappa\in[8/3,4],~\tilde\kappa\in[4,8].
\end{equation}
As in question~\ref{q12} we expect to have in the first regime a discrete holomorphic MO
corresponding to (\ref{eq:mo},\ref{eq:mobvp}) with $\sigma=-\beta/2=\frac3\kappa-2$,
whose massive counterpart corresponds to perturbation of $x\approx x_c$.
Then tending mass to $\infty$ we should arrive to the second critical regime.
Can we observe this effect in massive {SLE}{s}?
Namely, does such massive \makarovsmirnovicmpsle{\kappa}
tend to usual \makarovsmirnovicmpsle{\tilde\kappa} as $m\to\infty$? 
In particular, does massive \makarovsmirnovicmpsle3 tend to \makarovsmirnovicmpsle6?
Do other observables lead to different dualities?
Can we observe that tending mass to $-\infty$ leads to the frozen regime $x\in[0,x_c[$?
\end{question}

\begin{question}\label{q13}
Can massive observables (which are in some sense better behaved)
help to understand the critical ones?
This is the case in Ref.~\refcite{beffara-duminil-smirnov}, where we establish
the conjectured value of the critical temperature for FK model with $q\ge4$.
\end{question}

\section*{Acknowledgments}

This research was supported by the  N.S.F. Grant No. 0201893,
by the Swiss N.S.F.  and by the European Research Council AG CONFRA.


\begin{thebibliography}{10}

\bibitem{cft-collection}
C.~Itzykson, H.~Saleur and J.-B. Zuber (eds.), {\em Conformal invariance and
  applications to statistical mechanics} (World Scientific Publishing Co. Inc.,
  Teaneck, NJ, 1988).

\bibitem{schramm-icm}
O.~Schramm, Conformally invariant scaling limits: an overview and a collection
  of problems, in {\em Proceedings of the International Congress of
  Mathematicians (Madrid, August 22-30, 2006)\/},  (Eur. Math. Soc., Z\"urich,
  2007) pp. 513--543.

\bibitem{smirnov-icm1}
S.~Smirnov, Towards conformal invariance of 2{D} lattice models, in {\em
  Proceedings of the International Congress of Mathematicians (Madrid, August
  22-30, 2006)\/},  (Eur. Math. Soc., Z\"urich, 2006) pp. 1421--1451.

\bibitem{bauer-bernard-kytola}
M.~Bauer, D.~Bernard and K.~Kyt{\"o}l{\"a}, {\em J. Stat. Phys.} {\bf 132}, 721
  (2008).

\bibitem{bauer-bernard-cantini}
M.~Bauer, D.~Bernard,  and L.~Cantini, {\em J. Stat. Mech. Theory Exp.} ,
  P07037, (electronic) (2009).

\bibitem{nolin-werner}
P.~Nolin and W.~Werner, {\em J. Amer. Math. Soc.} {\bf 22}, 797 (2009).

\bibitem{garban-pete-schramm}
C.~Garban, G.~Pete and O.~Schramm, Work in progress (2008).

\bibitem{makarov-smirnov-mass}
N.~Makarov and S.~Smirnov, Massive stochastic {L}oewner evolutions,~ In
  preparation (2009).

\bibitem{dubedat-partition}
J.~Dub\'edat, S{LE} and the free field: {P}artition functions and couplings,
Preprint (2007).

\bibitem{schramm-sheffield-harmonic}
O.~Schramm and S.~Sheffield, {\em Ann. Probab.} {\bf 33}, 2127 (2005).

\bibitem{schramm-sheffield-dgff}
O.~Schramm and S.~Sheffield, {\em Acta Math.} {\bf 202}, 21 (2009).

\bibitem{MR2044671}
G.~F. Lawler, O.~Schramm and W.~Werner, {\em Ann. Probab.} {\bf 32}, 939
(2004).

\bibitem{smirnov-fk1}
S.~Smirnov, Conformal invariance in random cluster models, {I},
{\em Ann. of Math. (2)}  (to~appear).

\bibitem{riva-cardy-hol}
V.~Riva and J.~Cardy, {\em J. Stat. Mech. Theory Exp.} , P12001,
  (electronic) (2006).

\bibitem{chelkak-smirnov-iso}
D.~Chelkak and S.~Smirnov, Universality in the 2{D} {I}sing model and conformal
  invariance of fermionic observables,~ Preprint (2009).

\bibitem{chelkak-smirnov-spin}
D.~Chelkak and S.~Smirnov, Conformal invariance of the 2{D} {I}sing model at
  criticality,~ Preprint (2009).

\bibitem{smirnov-fk2}
S.~Smirnov, Conformal invariance in random cluster models, {II},
Preprint  (2009).

\bibitem{rogers-williams-2}
L.~C.~G. Rogers and D.~Williams, {\em Diffusions, {M}arkov processes, and
  martingales. {V}ol.~2}
(Cambridge University Press, Cambridge, 2000).

\bibitem{oksendal}
B.~{\O}ksendal, {\em Stochastic differential equations,}  
(Springer-Verlag, Berlin, 2003).

\bibitem{kager-nienhuis}
W.~Kager and B.~Nienhuis, {\em J. Statist. Phys.} {\bf 115}, 1149 (2004).

\bibitem{0985.60090}
S.~Smirnov, {\em C. R. Math. Acad. Sci. Paris} {\bf 333}, 239 (2001).

\bibitem{beffara-duminil-smirnov}
V.~Beffara, H.~Duminil-Copin and S.~Smirnov, 
The self-dual point of the 2{D} random-cluster model is
critical above $q=4$,
Preprint  (2009).

\end{thebibliography}
\end{document}